\documentclass[a4paper,11pt]{article}

\usepackage{a4wide}
\usepackage[english]{babel}

\usepackage{amsfonts}
\usepackage{amsmath}
\usepackage{graphicx}
\usepackage[colorlinks,bookmarks=true]{hyperref}
\usepackage{amsthm}

\newtheorem{lemma}{Lemma}[section]
\newtheorem{definition}{Definition}[section]
\newtheorem{theorem}{Theorem}[section]
\newtheorem{proposition}{Proposition}[section]
\newtheorem{remark}{Remark}[section]

\newtheorem{assumption}{Assumption}
\newtheorem*{theorem*}{Theorem}
\newtheorem*{proposition*}{Proposition}

\newcommand{\E}{\mathbb{E}}

\title{Market impact as anticipation of the order flow imbalance}
\author{Thibault Jaisson\thanks{I thank my PhD supervisors Emmanuel Bacry and Mathieu Rosenbaum for their help in the realization of this work.}\\ CMAP, \'Ecole Polytechnique Paris\\ thibault.jaisson@polytechnique.edu}

\begin{document}

\maketitle

\begin{abstract}
\noindent
In this paper, we assume that the permanent market impact of metaorders is linear and that the price is a martingale. Those two hypotheses enable us to derive the evolution of the price from the dynamics of the flow of market orders. For example, if the market order flow is assumed to follow a nearly unstable Hawkes process, we retrieve the apparent long memory of the flow together with a power law impact function which is consistent with the celebrated square root law. We also link the long memory exponent of the sign of market orders with the impact function exponent. One of the originalities of our approach is that our results are derived without assuming that market participants are able to detect the beginning of metaorders.
\end{abstract}

\noindent
\textbf{Keywords:}
Market impact,
metaorder,
market order flow,
Hawkes processes,
long memory,
square root law,
market efficiency.

\section{Introduction}

Loosely speaking, the \textit{market impact} is the link between the volume of an order (either market order or metaorder\footnote{A metaorder is a large transaction executed incrementally, see \cite{efficiencyimpact}.}) and the price moves during and after the execution of this order, see \cite{priceimpact}. Accurate modelling of market impact is of paramount importance. Practically, market impact greatly affects brokers, market makers but also large investors in the design of their optimal strategies. For example, market impact tends to decrease the profits of investment strategies by creating an execution cost that increases with the volume of the transaction. From a theoretical point of view, it is argued in \cite{hmsdcsd} and \cite{shiller} that endogenous (that is independent of information outside of the market) changes in supply and demand have more influence on price fluctuations than exogenous information. Therefore, understanding the price formation process necessarily goes through understanding market impact.\\

\noindent
Here we focus on metaorders. More precisely, we study the impact function of metaorders, which is the expectation of the price move with respect to time during and after the execution of the metaorder. We call permanent market impact of a metaorder the limit in time of the impact function (that is the average price move between the start of the metaorder and a long time after its execution). The literature on this topic is often rooted to the seminal work of Kyle \cite{kyle}. In the latter, it is shown that an informed trader who has private information on the future of the price should incrementally execute his metaorder, thereby slowly revealing the price to the rest of the market. However, it is also deduced that the market impact should be linear throughout the execution of the metaorder, which does not agree with empirical studies that consistently give a strictly concave market impact. We refer for example to \cite{bershova2013non} and \cite{moroimpact} where the estimated impact function is close to a power law with respect to time throughout the execution, with exponent being between 0.5, see \cite{bershova2013non}, and 0.7, see \cite{moroimpact}. These empirical results are often referred to as the square root law.\\

\noindent
The first main theoretical explanations of the concavity of the impact function connect it to the behaviour of large investors, see \cite{gabaix2006institutional}, the persistence in the order flow, see \cite{propagator} and \cite{lillolm}, or the size distribution of metaorders, see \cite{efficiencyimpact}.
In \cite{tothanomalous}, Toth \textit{et al.} show that the diffusivity of the price can be linked to the V-shape of the latent order book. They propose a model in which they numerically show that the concavity of the impact function depends on the long memory of the order flow. Recently, Farmer \textit{et al.} \cite{efficiencyimpact} have managed to relate this concavity to the distribution of the sizes of metaorders. More precisely, in their model, they get a power law impact function whose exponent $\nu$ is linked to the exponent of the distribution tail of the size of metaorders $\beta$ by $\nu=\beta-1$, which quite fits empirical data. To obtain this, they use two conditions: A martingale hypothesis on the price, and a fair pricing condition, which is derived from a Nash equilibrium between investors and states that the average ex post gain of a metaorder of any size should be zero. Let us remark that in \cite{donier}, similar conditions are obtained using only perfect competition between market makers.\\

\noindent
In the works mentioned above, it is assumed that market makers can precisely detect when a metaorder is being executed. Although reasonable for large metaorders, it is a strong postulate, especially at the beginning of the execution and for small metaorders. Here, we try to build a theory which does not require such an assumption. Indeed, market makers only see a flow of market orders and should thus set their prices according to this flow, and not to the underlying flow of metaorders.\\

\noindent
It is well known, see \cite{lillolm}, that the sign (buy or sell) of market orders presents persistence (more precisely, the correlation function of the sign process behaves as a power law with exponent lower than one). This persistence can be explained by the fact that the market order flow reflects the partition of metaorders into sequences of same-sided market orders. Indeed, it is shown in \cite{fragmentation} that an order flow obtained from the fragmentation of metaorders exhibits long memory. In \cite{fragvsherd}, it is empirically proved that the persistence in the sign process mostly comes from splitting (and not herding).
\\

\noindent
In \cite{propagator}, Bouchaud \textit{et al.} propose a way to obtain a diffusive price process from a long memory order flow modelled as a discrete time FARIMA process (see \cite{beranlm} for definition). To do so, they consider that the impact of a market order decreases in time as a power law (where the exponent of the power law is linked to the exponent of the long memory). However, in this work, the impact is not related to the notion of metaorder. In particular, if one executes a metaorder in this model, its permanent impact is null, see \cite{hmsdcsd}. Moreover, while the permanent impact of metaorders is not easy to compute, it does not seem to be zero in practice, see \cite{moroimpact}.
Our aim is to somehow place ourselves between \cite{propagator} and \cite{efficiencyimpact}. We assume that the permanent impact of metaorders does exist with a specific (linear) form derived from no price manipulation arguments. Furthermore, we consider that market makers only see the flow of market orders that they essentially understand as a superposition of fragmented metaorders.\\

\noindent
There are two manners to understand why metaorders impact the price. In most papers, see for example \cite{efficiencyimpact} and \cite{kyle}, market impact is viewed as a way to pass on private information to the price. In these models, large investors react to information signals on the future expectation of the price using metaorders. In such approaches, metaorders reveal ``fundamental" price moves but do not really cause them. In particular, if a metaorder is executed for no reason, it does not have any long term impact on the price. Here, we use the other vision of market impact. We assume that it is mechanical in the sense that a metaorder moves the price through its volume, by creating a long term imbalance in supply and demand, independently of the informativeness of the metaorder\footnote{In Section \ref{im}, we illustrate this vision of market impact with a simple investor model.}. Choosing between these two paradigms requires brokerage data where the client accepts to say whether he is trading because he has directional views on the price or for another reason (risk management, hedging, regulatory constraints,...). This is done in \cite{waelbroeck2013market} where it is shown that after a few days, the impact of ``informed trades" is larger than that of ``cash flow trades" (which even tends to be null). However, at intraday time scales, which are of interest here, the two previous impact functions are very close.\\

\noindent
In some papers, see for example \cite{gatheralna} and \cite{manip}, the mechanical vision of market impact is implicitly assumed to derive important results. In these two articles, the authors derive conditions in their models so that a round trip\footnote{A round trip is defined in \cite{manip} as a trading strategy which starts and ends with a null inventory.} is not profitable on average. In \cite{manip}, it is applied to the Almgren-Chriss optimal liquidation framework, see \cite{AC}, and the authors obtain that the permanent impact of a trade must be linear in its size. In \cite{gatheralna}, this hypothesis is considered in a more general model and once again, it is obtained that if there is permanent impact, it has to be linear. However, these reasonings only hold if the impact of a trade does not depend on the reason investors are executing their trades, therefore, under the hypothesis that the impact is mechanical.\\

\noindent
In this work, we make two natural hypotheses: The linearity of the permanent market impact of a metaorder (which arises when impact is mechanical and there is no price manipulation but can also hold if impact is informational, see Section 2), and a martingale assumption on the price. These two hypotheses give us a simple and general relationship between the dynamics of the order flow and the evolution of the price (Equation \eqref{impactvolume}). Let us emphasize that this relationship is not between the price and the flow of specific metaorders which are not observable by market makers. We then apply this formula to the example of an order flow modelled by a nearly unstable Hawkes process.\\

\noindent
Hawkes processes, introduced in \cite{hawkes1}, have been successfully applied to seismology, neurophysiology, epidemiology, and reliability and are nowadays very popular in finance. Among their recent applications in this field, let us cite the studies of tick by tick prices in \cite{bacryhawkes1}, trades in \cite{bowsher2007modelling}, order books in \cite{large2007measuring}, financial contagion in \cite{ait2010modelling} and credit risk in \cite{errais2010affine}. Using Hawkes processes to model the flow of market orders in continuous time is natural. Indeed, it is a simple way to introduce clustering in the order flow and to reproduce the empirical correlation functions of buy and sell orders, see \cite{bacryhawkesesti}.\\

\noindent
We will see that we recover many well known stylized facts of market impact under the Hawkes assumption. In particular, at the scale of individual market orders, we will compute a continuous time propagator model which generalizes the discrete time propagator proposed in \cite{propagator}. In such models, the impact of a market order on the price does not depend on the past of the market order flow but is decreasing in time. At the scale of metaorders, we will compute a power law impact function, where the exponent of the power law $\nu$ is linked to the exponent of the long memory of the order flow $\gamma$ by the relation: $\gamma=2\nu-1$. We thus somehow recover the square root law.\\

\noindent
The paper is organized as follows. In Section \ref{agentmodel}, we explain and justify the use of the linear permanent impact assumption and we compute our general impact equation. In Section \ref{example}, we apply this equation to the example of an order flow modelled by nearly unstable Hawkes processes. We derive the price process implied by such an order flow and show that the impact function of a metaorder is asymptotically a power law with exponent lower than one. We conclude in Section \ref{conclusion}.

\section{Computing price moves from the order flow}
\label{agentmodel}

The aim of this section is to present our assumptions and to see how they imply our impact equation that links the dynamics of the price to that of the order flow.

\subsection{Linear permanent impact under the mechanical impact assumption}
\noindent
In this section, we explain why it is reasonable to assume that the permanent market impact ($PMI$ for short) of a metaorder is linear in its volume under the mechanical vision of market impact.\\

\noindent
During and after the execution of a metaorder of volume $V$, the price process is impacted. Indeed, if one buys a large quantity of stock, the price moves upward on average. We assume that a long time after the execution of the metaorder (compared to the time it takes to execute it), the market somehow returns to some kind of stationary state except that the price has on average moved by a value that only depends on $V$. The permanent market impact, is by definition this average price move between the start of the execution of the order and a long time after the end of this execution:

\begin{center}
$PMI(V)=\displaystyle\lim_{s\rightarrow+\infty}\E[P_s-P_0|V],$
\end{center}
where $s=0$ corresponds to the beginning of the metaorder and $P_s$ is the mid price process. We assume that the impact of metaorders is large enough so that we can neglect the spread and thus the mid, ask and bid prices are close.\\

\noindent
In order to measure this impact, one thus has to wait for the market to absorb the volume of the order and to get to its new equilibrium.

\begin{remark}
The notion of expectation is in fact not defined outside of a model. Here, we make the assumption (that is implicit in most works on market impact) that the dynamics of the market satisfy stationarity conditions which allow us to compute expectations as empirical averages.
\end{remark}

\begin{remark}
The need to wait a long time (metaorders can last days), the variations of the market state (market impact depends on the stock, volatility, market activity,...) as well as the relative small number of recorded metaorders (in \cite{moroimpact}, they detect around 100 000 metaorders per year on 74 stocks traded on the LSE, in \cite{bershova2013non} they use 12 500 metaorders) make it very difficult to have robust empirical laws on the permanent impact. For example, in \cite{bershova2013non}, the empirical permanent impact of a metaorder with respect to its duration is derived by averaging over metaorders on many stocks and market states (time periods) which might seem arguable. We can however say that the permanent market impact seems different from zero.
\end{remark}

\noindent
We place ourselves in the mechanical impact paradigm explained in the introduction and give two simple possible reasons why this impact should be linear in the volume of the metaorder:

\begin{center}
$PMI(V)=kV.$
\end{center}

\noindent
The first one is functional. Indeed, if this impact does not depend on the way the investor executes his order, then the two following executions should lead to the same permanent impact:

\noindent
i) Executing a volume $V$.

\noindent
ii) Executing a volume $\rho V$, waiting a long time (for the market to return to its stationary state) and executing $(1-\rho)V$ ($\rho\in[0,1]$).

\noindent Therefore,
\begin{center}
$PMI(V)=PMI(\rho V)+PMI((1-\rho)V),$
\end{center}
which implies that $PMI$ must be linear.

\begin{remark}
This reasoning does not apply to the temporary (during the execution) part of the impact function which may depend on the way the metaorder is executed. Indeed, in ii), one has to wait a long time after the first part of the order has been executed otherwise the market would not be in its stationary state for the second part.
\end{remark}

\noindent
The second reason is based on ``no price manipulation" arguments. In \cite{manip}, it is shown that in the framework of Almgren and Chriss \cite{AC}, where the impact is decomposed into execution costs and permanent market impact, the absence of price manipulations implies that the permanent market impact must be linear.\\

\noindent
In \cite{gatheralna}, these arguments are generalized to a more realistic time dependent impact model. The model states that if a trader executes an order at a trading rate $\stackrel{.}{x}_s$ at time $s$ (assuming that the volume $x_t$ executed up to time $t$ is differentiable), then he impacts the price $S$ in $t$ by: $\int_0^t f(\stackrel{.}{x}_s)G(t-s)ds$. Therefore, if without the trader, the price is a Brownian motion with volatility $\sigma$, the price satisfies:

\begin{equation*}
S_t=S_0+\int_0^t f(\stackrel{.}{x}_s)G(t-s)ds+\int_0^t \sigma dB_s,
\end{equation*}
where $f$ is the impact function and $G$ is the decay kernel, which represents the decrease of the impact of a trade through time.\\

\noindent
In this model, using the same kind of arguments as in \cite{gatheralna}, we can easily show that if there is permanent impact, then this impact is linear in volume. More precisely, we have the following proposition:

\begin{proposition}
\label{propogatheral}
If there is no price manipulation in the sense of \cite{gatheralna}\footnote{In \cite{gatheralna}, a price manipulation is a round trip whose average cost is negative.} and there is permanent impact, that is $G(s)\underset{s\rightarrow +\infty}{\rightarrow} G_\infty>0$, then $f$ is linear.
\end{proposition}

\noindent
The proof is given in appendix.

\begin{remark}
In the previous model, the cumulated executed volume of the trader $x_t=\int_0^t \stackrel{.}{x}_s ds$ is assumed to be differentiable which is not the case in practice. Indeed, the volume is executed by individual market orders. However, if $f(v)=kv$ as above, the model can be rewritten:
\begin{equation*}
S_t=S_0+k\int_0^t G(t-s)dx_s+\int_0^t \sigma dB_s,
\end{equation*}
and generalised to real order flows. The permanent impact of a metaorder of volume $V$ is thus linear in its volume:$$PMI(V)=G_\infty k V.$$

\end{remark}

\noindent
Therefore, in a rather general model, the permanent impact of a metaorder has to be linear. However, let us point out that the preceding arguments assume that even traders who execute metaorders without information have a permanent impact on the price. Therefore, the impact has to be mechanical as explained in the introduction. This could seem unrealistic but in fact can be understood in terms of the aversion that asset managers have to holding large positions of an asset.
In the next paragraph, we present a simple investor model which explains how even uninformed metaorders can impact an indifference price.

\subsection{A toy investor model}
\label{im}
A classical view is to consider that investors who use metaorders have a better information about the future of the price than other market participants, see for example \cite{kyle}. However, this is based on the economical concept of an ``informational price" whereas empirical studies seem to show that the main drive of price fluctuations is not information, see \cite{shiller}, but volume, see \cite{hmsdcsd}. Indeed, in \cite{shiller}, Shiller argues that: ``The movements in stock price indexes could not realistically be attributed to any objective new information, since movements in the indexes are ``too big" relative to actual subsequent-movements in dividends".\\

\noindent
In this section, we consider a simple example of investor model, which is in fact a particular case of the CAPM model (see for example \cite{sharpe1964capital}). This model enables us to illustrate how an uninformed metaorder can have a long term mechanical impact on the price. This impact is due to the risk aversion of asset managers (that is their unwillingness to hold large positions of one asset) and not to information asymmetry.\\

\noindent
We assume that there are $n$ investors in our market.
There is a fixed number $N$ of shares that are spread between the agents and let us call $P$ the price.
Every investor $(I_i)_{i=1...n}$ estimates that the distribution of the yield of a share has an average $E_i$ and a variance $\Sigma_i$ and chooses the number of shares $N_i$ in his portfolio optimizing a mean variance criterion with risk aversion $\lambda_i$:
\begin{eqnarray*}
N_i &=& \text{argmax}_x \{ x(E_i-P)-\lambda_i x^2\Sigma_i\}\\
&=&\frac{E_i-P}{2\lambda_i \Sigma_i}.
\end{eqnarray*}

\noindent
Moreover,
$$\sum_{i=1}^n N_i=N,$$
which gives the indifference price $P$ of investors:
$$P=\frac{\displaystyle\sum_{i=1}^n\frac{E_i}{2\lambda_i \Sigma_i}-N}{\displaystyle\sum_{i=1}^n\frac{1}{2\lambda_i \Sigma_i}}.$$

\noindent
Let us now assume that the total number of shares becomes $N-N_0$ due to the action of some non optimizing agent needing to buy some shares (for cash flow reasons for example).
The new indifference price is
$$P^+=P+\frac{N_0}{\displaystyle\sum_{i=1}^n\frac{1}{2\lambda_i \Sigma_i}}=P+kN_0.$$

\noindent
This means that in this framework, a metaorder has an impact on the indifference price that is linear in its size. We thus have illustrated (in a quite naive way admittedly!) how the impact of a trade might not depend on its informativeness.

\begin{remark}
Let us stress again that in contrast to \cite{efficiencyimpact} and \cite{kyle}, in our framework, the impact of an order is purely due to its volume and not to any exogenous information. Therefore, all traders are equivalent and there are no noise traders. In particular, individual trades should be considered as small metaorders.
\end{remark}

\subsection{Impact dynamics}

We have shown that it is reasonable to consider that every metaorder executed by a market participant has a long term impact on the price that is linear in its size. Here, we moreover assume a martingale property for the price. We use these assumptions to get a general impact equation linking the order flow dynamics and the price process.\\

\noindent
Let us denote by $[0,S]$ the time during which metaorders are being executed (which can be thought of as the trading day) and by $V^a_t$ and $V^b_t$ the cumulated volumes of market orders at the ask and at the bid up to a time $t\in[0,S]$. We consider the following assumption.

\begin{assumption}
\label{assu1}
Metaorders are executed only via market orders and market orders are only used to execute metaorders. Therefore, the daily cumulated volume of buy (resp. sell) market orders $V^{a}_S$ (resp. $V^{b}_S$) is equal to the sum of the volume of the $N_S^a$ buy (resp. $N_S^b$ sell) metaorders of the day
\begin{equation*}
V^{a}_S=\sum_{i=1}^{N_S^a} v^a_i,
\end{equation*}
where $v^a_i$ is the volume of the $i^{th}$ buy metaorder of the day.
\end{assumption}

\begin{remark}

Assumption \ref{assu1} is used for simplicity and can in fact be weakened to Assumption $1$ of \cite{mminvariants} that states that the daily cumulated volume of buy (resp. sell) market orders  is proportional (and not necessarily equal) to the sum of the volume of the $N_S^a$ buy (resp. $N_S^b$ sell) metaorders of the day:
\begin{equation*}
V^{a}_S=\frac{\chi}{2}\sum_{i=1}^{N_S^a} v^a_i,
\end{equation*}
where  $\chi$ can be understood as \textit{``the level of intermediation in the market"}. Hereinafter, we will take $\chi=2$.
\end{remark}

\noindent
To fix ideas, we assume that at the beginning and at the end of the day, there are auctions whose prices are denoted by $P_0^-$ and $P_S^+$.
Then our assumption on the linearity of the permanent market impact goes as follows.

\begin{assumption}
\label{assu2}
$P_S^+$ is on average equal to the price of the opening auction $P_0^-$ plus the impact of the metaorders of the day which is taken linear in their volume:
\begin{eqnarray*}
P_S^+&=&P_0^-+\kappa (\sum_{i=1}^{N_S^a} v^a_i-\sum_{i=1}^{N_S^b} v^b_i)+Z_S\\
&=&P_0^-+\kappa (V^a_S-V^b_S)+Z_S,\\
\end{eqnarray*}
where $\kappa$ is a positive constant and $Z$ is a martingale random walk.
\end{assumption}

\noindent
Finally, we consider the following assumption.

\begin{assumption}
\label{assu3}
The price $P$ is a martingale.
\end{assumption}

\noindent
Assumption \ref{assu3} implies that
\begin{eqnarray*}
P_t&=&\E[P_S^+|\mathcal{F}_t]\\
&=&P_0^-+\E[\kappa(V^a_S-V^b_S)|\mathcal{F}_t]+\E[Z_S|\mathcal{F}_t],
\end{eqnarray*}
where $\mathcal{F}_t$ is the information set of market participants at time $t$.\\

\noindent
The noise term $\E[Z_S|\mathcal{F}_t]$ corresponds to price moves that are caused by other factors than volume (for example public announcements). For the sake of simplicity, we consider that $Z_S=0$ and thus

\begin{center}
$P_t=P_0^-+\E[\kappa(V^a_S-V^b_S)|\mathcal{F}_t].$
\end{center}

\noindent
Let us now extend the buy and sell market order flow processes to $\mathbb{R_+}$ (until now they were only defined on a trading day $[0,S]$) and assume that these processes satisfy the following technical assumption.

\begin{assumption}
\label{assu4}
The quantity $\E[V^a_s-V^b_s|\mathcal{F}_t]$ converges to some finite limit when $s$ tends to infinity.
\end{assumption}

\noindent
This assumption can be understood in this way: $\mathcal{F}_t$ does not provide any information about the order flow imbalance which will occur at time $t'\gg t$. Indeed, under this assumption, for all $h>0$,

\begin{center}
$\E[(V^a_{s+h}-V^b_{s+h})-(V^a_s-V^b_s)|\mathcal{F}_t]\underset{s\rightarrow +\infty}{\rightarrow} 0$.
\end{center}

\noindent
Therefore, the order flow imbalance between $s$ and $s+h$ is asymptotically (in $s$) not predictable at time $t$.
For example, in the next section, we will assume that the market order flows are Hawkes processes for which this assumption is satisfied.
Under this hypothesis, we have that if $S$ is large enough compared to the characteristic convergence scale of the function $s\mapsto \E[V^a_s-V^b_s|\mathcal{F}_t]$, when computing the price $P_t$, $S$ can be seen as infinity. This implies the main result of this paper.

\begin{theorem}
Under Assumptions \ref{assu1}, \ref{assu2}, \ref{assu3} and \ref{assu4}

\begin{equation}
\label{impactvolume}
P_t=P_0+\kappa \displaystyle\lim_{s\rightarrow+\infty}\E[V^a_s-V^b_s|\mathcal{F}_t].
\end{equation}

\end{theorem}

\begin{remark}
To fix ideas, we have assumed that metaorders stop at the end of the day $S$. However, some metaorders last days (or even weeks). Therefore, $S$ should thus not be considered as the end of the day but the end of a period that is very large compared to the execution of metaorders.
\end{remark}

\noindent
From Equation \eqref{impactvolume}, and any example of order flow dynamics, we can thus derive a price process which satisfies the following properties:

\begin{itemize}
\item{The price is a martingale even if the order flow exhibits persistence.}
\item{The permanent impact of a metaorder is linear in its size, independently of the execution.}
\item{The price process only depends on the global market order flow and not on the individual executions of metaorders. We thus do not need to assume that the market ``sees'' the execution of metaorders as it is usually done (see for example \cite{efficiencyimpact}).}
\end{itemize}

\noindent
Equation \eqref{impactvolume} can be seen as a rigorous and model independent generalization of the postulate ``the impact is proportional to the innovation in the order flow" of \cite{glomil} and \cite{mrr}. Indeed, in our model, market orders move the price because they change the anticipation that market makers have about the future of the order flow.
Furthermore, considering, as in Equation \eqref{impactvolume}, that price moves are due to the ``surprise in the order flow" (that is the variation of the expected cumulated order flow imbalance) is a general way to solve the following apparent ``paradox" (which is one of the challenges of impact models): Order flows are persistent and yet prices are martingales.\\

\noindent
From the empirical point of view, our framework is supported by \cite{hirschey2011high} where it is shown that the trades of some arbitragers (namely high frequency traders) anticipate the order flows of other investors. Therefore, arbitragers move prices through their trades until it is no longer worth doing so; that is until the price correctly anticipates the future order flow imbalance.\\

\noindent
In the next section, we give an example of realistic order flow dynamics and compute the associated price process. We will see that the linearity of the permanent impact of a metaorder is not incompatible with the observed concave power law impact functions.

\section{Application to the Hawkes order flow example}
\label{example}

We now want to apply Equation \eqref{impactvolume} to a reasonable example of market order flow dynamics. More precisely, we model the buy and sell market order flows as two independent (nearly unstable) stationary Hawkes processes which somehow reproduce the persistence of the sign of market orders. Applying Equation \eqref{impactvolume} to these processes, we derive a simple price model that is very similar to the propagator model of \cite{propagator}.

\subsection{Impact of individual market orders}

Before deriving the impact of individual market orders under the Hawkes order flow assumption, we need to introduce Hawkes processes and the propagator model.

\subsubsection{Hawkes processes}

Hawkes processes have been introduced in \cite{hawkes1}. They have recently been applied to many fields of finance, see the references in introduction. For example, they were used to study tick by tick prices in \cite{bacryhawkes1} and order flows in \cite{bacryhawkesorder}. We model the arrival times of buy and sell market orders as the jump times of two independent Hawkes processes.

\begin{remark}
We consider that all market orders have the same volume $v$.
\end{remark}

\noindent
By definition, a Hawkes process $N$ is a self exciting point process whose intensity $\lambda$ at time $t$ depends linearly on its past, see \cite{hawkes1}:
$$\lambda_t=\mu+\int_0^t\phi(t-s)dN_s,$$
where $\mu$ is a positive constant and $\phi$ is a positive function supported on $\mathbb{R}_+$ which satisfies the stability condition: $$\int_0^{+\infty} \phi(s)ds <1.$$

\noindent
A nice property of Hawkes processes that we will use in order to compute the price dynamics is the following, see \cite{bacryhawkesscaling}:

\begin{proposition}
\label{formlambdamart}
$\forall t\geq 0$:
$$\lambda_t=\mu+\mu \int_0^t \psi(t-s)  ds+\int_0^t \psi(t-s)dM_s,$$
where $\psi=\sum_{k=1}^{+\infty} \phi^{*k}$, $\phi^{*k}$ being the $k^{th}$ covolution product of $\phi$ and $dM_t=dN_t-\lambda_tdt$.

\end{proposition}

\noindent
Therefore, since $M$ is a martingale and $\psi$ is supported on $\mathbb{R}_+$, for $0\leq s,r\leq t$,
$$\E[N_t-N_s|F_r]=\int_s^t \E[\lambda_u|F_r]du$$ and $$\E[\lambda_u|F_r]=\mu+\mu \int_0^u \psi(u-x) dx+\int_0^r \psi(u-x)dM_x$$

\noindent
This formula allows us to compute the anticipation of the order flow and will be very useful to apply Equation \eqref{impactvolume} in the Hawkes context.\\

\noindent
Hawkes processes are a natural extension of Poisson processes and can be used to model the empirical clustering of same-sided trades.
More precisely, we will later show that a Hawkes order flow can almost reproduce the empirical long memory of the signed flow of market orders.

\subsubsection{Propagator model}

In \cite{propagator}, the order flow is modelled as a FARIMA process that enables to obtain the long memory property of this flow observed in practice. In order to have a diffusive price, the authors propose a price impact model called propagator model where the impact of a market order does not depend on the past of the order flow but is transient (its permanent market impact is null). Proceeding essentially as in \cite{gatheralfred}, this model can be extended to continuous time the following way:

\begin{definition}
A continuous time propagator model is an impact model where the price process $P$ is linked to the cumulated signed volume of market orders $V=V^a-V^b$ by:
$$P_t=P_0+\int_0^t \zeta(t-s) dV_s,$$
where $\zeta$ is a decay kernel.
\end{definition}

\begin{remark}
There is a small difference between this model and the model of \cite{gatheralfred}: In \cite{gatheralfred}, $V$ is the trading rate of a given trader while here it is the order flow of all the market.
\end{remark}

\subsubsection{Computation of the price dynamics}

From now on, we model the arrivals of buy and sell market orders as the points of two independent Hawkes processes $N^a$ and $N^b$ with identical kernel $\phi$ and exogenous intensity $\mu$. We also assume that the volume of each market order is constant equal to $v$.
\begin{proposition}
\label{theo2}
In the Hawkes order flow model, assuming that the price and volume satisfy the dynamics Equation \eqref{impactvolume}, then the price process follows a propagator model:

\begin{equation}
\label{propagator}
P_t=P_0+\int_0^t \zeta(t-s)(dN^a_s-dN^b_s),
\end{equation}
with $\zeta(t)=\kappa v(1+\int_t^{+\infty}\psi(u)-\int_0^t\psi(u-s)\phi(s)dsdu)$.
\end{proposition}

\noindent
The proof is given in appendix.\\

\noindent
Let us notice that the kernel does not tend to 0 since there is permanent impact.\\

\noindent
Our propagator kernel $\zeta$ compensates the correlation of the order flow implied by the Hawkes kernel $\phi$ to recover a martingale price.

\subsubsection{Properties of our propagator model}

In this paragraph, we place ourselves in a general framework which includes the result of Proposition \ref{theo2} and generalizes it to critical Hawkes order flows (with $\int \phi =1$, see \cite{hawkescritic} for existence). We show that a continuous time propagator price model combined with a Hawkes order flow can lead to a martingale price under a simple condition on $\zeta$ and $\phi$. We then use this condition to obtain a nicer expression of $\zeta$. Indeed, if:

\begin{itemize}
\item{The buy and sell market order flows are two independent Hawkes processes with kernel $\phi$:
\begin{center}
$\lambda^{a/b}=\mu+\int_0^t\phi(t-s)dN^{a/b}_s.$
\end{center}
}
\item{The price process follows a propagator model:
\begin{center}
$P_t=P_0+\int_0^t\zeta(t-s)(dN^a_s-dN^b_s).$
\end{center}
}
\end{itemize}

\noindent
Then, 
\begin{itemize}
\item{If there is no trade in $t$, the price is differentiable in $t$ with derivative
$$P'_t=\int_0^t\zeta'(t-s)(dN^a_s-dN^b_s).$$}
\item{If there is a buy trade in $t$ (which happens with intensity $\lambda_t^a=\mu+\int_0^t\phi(t-s)dN^{a}_s$), the price jumps upward of $\zeta(0)$.}
\item{If there is a sell trade in $t$ (which happens with intensity $\lambda_t^b=\mu+\int_0^t\phi(t-s)dN^{b}_s$), the price jumps downward of $\zeta(0)$.}
\end{itemize}
Therefore,
\begin{center}
$\displaystyle\lim_{h\rightarrow 0} \frac{E[P_{t+h}|\mathcal{F}_t]-P_t}{h}=\zeta(0)\int_0^t\phi(t-s)(dN^a_s-dN^b_s)+\int_0^t\zeta'(t-s)(dN^a_s-dN^b_s).$
\end{center}

\noindent
Therefore, we have the following result:
\begin{proposition}
In the framework defined above, if the price is a martingale then:
\begin{equation}
\label{phi-zeta}
\zeta'(x)=-\zeta(0)\phi(x).
\end{equation}

\noindent
Given a function $\zeta$, we can compute a unique $\phi_\zeta$ which satisfies \eqref{phi-zeta}. Conversely, given a $\phi$ and a $\zeta(0)$ (or $\zeta$ anywhere else), we can compute a unique $\zeta_\phi$ which satisfies \eqref{phi-zeta}.
\end{proposition}

\begin{remark}
Let us stress that in the general framework of this paragraph, the Hawkes order flows can be critical ($\int\phi=1$). Moreover, the criticality of the Hawkes process is equivalent to the transience of the impact of a market order:
\begin{center}
$\displaystyle\lim_{x \rightarrow +\infty}\zeta(x)=0 \Leftrightarrow \int\phi=1.$
\end{center}

\noindent
In this critical case, Assumption \ref{assu4} is not satisfied. We are thus outside of the framework of Equation \eqref{impactvolume}.
\end{remark}

\noindent
Going back to the framework of Proposition \ref{theo2} (in which $\int \phi <1$), since our price $P_t$ is an expectation conditionally on $\mathcal{F}_t$, it is by construction a martingale and therefore, the implied $\zeta$ must satisfy \eqref{phi-zeta}. Indeed, we have:
\begin{eqnarray*}
\zeta'(t)&=& \kappa v[-\psi(t)+\int_0^t\psi(t-s)\phi(s)ds-\int_t^{+\infty}\psi(u-t)\phi(t)dy]\\
&=&\kappa v[-\psi(t)+\psi\ast\phi(t)-\phi(t)|\psi|]\\
&=&-\kappa v(|\psi|+1)\phi(t)
\end{eqnarray*}
since $$\psi\ast\phi=\psi-\phi$$ and thus
$$\zeta'(t)=-\zeta(0)\phi(t).$$

\begin{theorem}
In the setting of Proposition \ref{theo2}, we can get another expression for $\zeta$:
\begin{equation}
\label{zeta2}
\zeta(x)=\zeta(0)(1-\int_0^x\phi(s)ds)=\kappa v\frac{1}{1-\int \phi }((1-\int \phi )+\int_x^{+\infty} \phi(s)ds).
\end{equation}
\end{theorem}

\noindent
The ratio between the temporary and the permanent impact of an order satisfies
$$\lim\zeta/\zeta(0)=1-\int \phi .$$
Therefore, when $\int \phi $ is close to one which, we will see, corresponds to the apparent long memory asymptotic, the temporary impact is much larger than the permanent impact. This is due to the fact that in this asymptotic a market order on average ``generates'' a lot of market orders, see next paragraph.

\subsubsection*{Cluster representation of Hawkes processes}

In this paragraph, we recall the branching construction of Hawkes processes which will allow us to interpret $\int \phi$ in terms of market impact.\\

\noindent
Let us define a population model (see the Introduction of \cite{hawkescritic}): At time zero, there are no individuals. Some individuals (migrants) arrive as a uniform Poisson process with intensity $\mu$. If a migrant arrives at time $s$, the birth dates of its children form a Poisson process of intensity $\phi(.-s)$, with $\int \phi <1$. In the same way, if a child is born in $s'$, the birth dates of its children form a Poisson process of intensity $\phi(.-s')$. We call $N_t$ the number of individuals who were born or migrated until $t$. We have that $N$ is a point process of intensity
$$\lambda_t=\mu+\int_0^t\phi(t-s)dN_s,$$
and therefore, $N$ is a Hawkes process.\\

\noindent
If we now name a cluster the set of all individuals who descend from the same migrant, a Hawkes process is the superposition of independent clusters which arrive at a rate $\mu$. This is the Poisson cluster representation of Hawkes processes, see \cite{hawkescluster}. Replacing ``individual" by ``market order" and ``cluster" by ``metaorder", a Hawkes process can be seen as the superposition of independent executions of metaorders. In this interpretation, the average size of a metaorder is $1/(1-\int \phi )$. Therefore, when $\int \phi $ is close to one, a metaorder is on average made of a large number of market orders. The parameter $\int \phi $ can thus be seen as the degree of endogeneity of the market and has recently been the subject of many investigations, see \cite{filimonov2013apparent}, \cite{hardiman2013critical} and \cite{jaisson2013limit}.

\subsection{Impact of metaorders}

In this paragraph, we study the impact of a large group of trades that we see as a metaorder in an asymptotic framework. To do so, we assume that the Hawkes order flows are nearly unstable (their kernel's norm is close to one) which enables to retrieve the persistence property which is observed on the data.

\subsubsection{Apparent long memory}
\label{alm}

It is well known that the sign of market orders presents long memory: Its autocorrelation function asymptotically behaves as a power law with exponent lower than one, see \cite{lillolm}. Recall that the Fourier transforms of the kernel $\phi$ and of the autocorrelation $Cov(.,h)$ of the increments of length $h$\footnote{$Cov(\tau,h)=\mathbb{E}[(N_{t+\tau+h}-N_{t+\tau})(N_{t+h}-N_{t})]-(\mathbb{E}[N_{t+h}-N_{t}])^2$.} of a stationary Hawkes process (denoted respectively by $\hat{\phi}$ and $\widehat{Cov}(.,h)$) are linked by:
$$\widehat{Cov}(z,h)=h\frac{\frac{\mu}{1-\int \phi }\hat{g}^h_z}{|1-\hat{\phi}(z)|^2},$$
where $\hat{g}^h_z$ is the Fourier transform of $g^h_\tau=(1-|\tau|/h)^+$, see \cite{bacryhawkesesti}.\\

\noindent
Therefore, since $$\widehat{Cov}(0,h)=\int_0^\infty Cov(\tau,h)d\tau$$ and $$\int\phi=\hat{\phi}(0),$$the infinity of the integral of the correlation function, which defines long memory, implies that the norm of the kernel must be equal to one. Such ``long memory" Hawkes processes have been studied in \cite{hawkescritic}.\\

\noindent
However, estimations performed in \cite{bacryhawkesorder}, \cite{filimonov2013apparent} and \cite{hardiman2013critical} seem to show that $\int \phi $ is close to one but strictly lower than one ($\int\phi\sim 0.9$). Such nearly unstable Hawkes processes have been studied in \cite{jaisson2013limit}, where it is shown that their diffusive limit is an integrated Cox-Ingersoll-Ross process. However, the latter work assumes that the shape of the kernel has an average in the sense that $\int x \phi(x) dx<+\infty$, which is not really in agreement with the power law persistence of the sign process. Qualitative arguments and numerical simulations presented in \cite{bacryhawkesorder} show that taking a nearly unstable Hawkes process with $\phi$ behaving asymptotically as $1/x^{1+\alpha}$ should lead to apparent long memory up to the time scale $1/(1-\int \phi )^{1/\alpha}$ (that is when the gap $\tau$ between the considered $h$-increments is much lower than $1/(1-\int \phi )^{1/\alpha}$). This is consistent with the empirical kernel estimations performed in the same paper.\\

\noindent
In order to give a formal version of this statement, we proceed as in \cite{jaisson2013limit} and consider a sequence of renormalized Hawkes processes whose kernels'norm tends to one ``faster" than the observation scale $T$ tends to infinity.

\subsubsection*{Asymptotic framework}

For every observation scale $T>0$, we define the stationary Hawkes process $N^T$ of exogenous intensity $\mu^T=C_\mu(1-a_T)T^{2\alpha-1}$ and of kernel $\phi^T=a_T\Phi$, where $\int \Phi=1$ and $(a_T)$ is a sequence of positive real numbers smaller than one which tends to one as $T$ tends to infinity. Furthermore, we assume that $$\Phi(x)\underset{x\rightarrow+\infty}{\sim}\frac{\alpha c^\alpha}{x^{1+\alpha}},$$
where $c>0$ and $\alpha\in]0,\frac{1}{2}[$.\\

\noindent
Since $T$ is the observation scale, we consider the process $N^T_{Tt}$ that we renormalise by considering $X^T_t:=A_T N^T_{Tt}$, with $A_T=(1-a_T)/(T\mu^T)=1/(T^{2\alpha})$.
Of course, $A_T$ and $\mu^T$ are chosen so that the expectation and the covariance of the sequence of processes converge.\\

\noindent
Considering the process at time scales below  $1/(1-\int \phi )^{1/\alpha}$ means that the observation scale $T$ is much smaller than $1/(1-a_T)^{1/\alpha}$. It can be formally written the following way:

\begin{assumption}
\label{scalelm}
$T(1-a_T)^{1/\alpha}\rightarrow 0.$
\end{assumption}

\subsubsection*{Result}

Let $h>0$ be fixed, and the covariance of the $h$-increments of $X^T$ be denoted by:
\begin{center}
$C^{X,T}_{\tau,h}:=Cov[(X^T_{t+\tau+h}-X^T_{t+\tau}),(X^T_{t+h}-X^T_{t})].$
\end{center}

\noindent
The following result holds.

\begin{theorem}
\label{theo1}
Under Assumption \ref{scalelm}:
\begin{equation}
\label{fbm}
C^{X,T}_{\tau,h} \rightarrow C^{X}_{\tau,h}=K(|\tau+h|^{2H}+|\tau-h|^{2H}-2|\tau|^{2H}),
\end{equation}
with $H=1/2+\alpha$ and $K$ is a constant explicitly given in the proof in appendix.\\

\noindent
This limiting function satisfies the long memory property: for any $h>0$,

\begin{center}
$C^{X}_{\tau,h}\underset{\tau\rightarrow +\infty}{\sim} \frac{C_h}{\tau^{1-2\alpha}}.$
\end{center}

\end{theorem}

\noindent
The proof is given in appendix.

\begin{remark}
Equation \eqref{fbm}, corresponds exactly to the correlation function of a fractional Brownian motion of Hurst index $H$.
\end{remark}

\begin{remark}
A more precise behaviour of such ``heavy tailed nearly unstable Hawkes processes'' is heuristically discussed in \cite{jaisson2013limit}.
\end{remark}

\noindent
In \cite{bacryhawkesorder}, it is shown that the time scale $1/(1-\int \phi )^{1/\alpha}$ up to which a Hawkes process can present apparent long memory can be rather large ($10^3$ seconds). This implies that modelling financial order flows as nearly unstable Hawkes processes is consistent with their empirical apparent persistence.

\subsubsection{Modelling metaorders}

In the above framework, $N^a$ and $N^b$ are the flows of anonymous market orders. This corresponds to modelling the market from the point of view of a passive agent who does not use orders and does not know who uses the different orders. In order to compute the impact function of metaorders, it is convenient to look at the market from the point of view of someone who is executing a (buy) metaorder, see \cite{bacryhawkesorder}. To do so, we consider that the total order flow is the sum of anonymous orders (which are modelled as two independent Hawkes processes as before) and labelled orders which correspond to ``our'' order flow.
More precisely, during the execution of a metaorder, we model our order flow as a Poisson process $P^{F,\tau}$ of intensity $F$ on $[0,\tau]$. The total buy and sell order flows are thus $N^a+P^{F,\tau}$ and $N^b$.\\

\noindent
Let us compute the impact function of a metaorder in our model. From the point of view of the rest of the market, there is no difference between our orders and anonymous orders. It thus seems natural to assume that the price impact of our orders is the same as the price impact of anonymous orders. Therefore, according to the propagator model:

\begin{center}
$P_t=P_0+\int_0^t \zeta(t-s)(dN^a_s-dN^b_s)+\int_0^t \zeta(t-s)dP^{F,\tau}_s.$
\end{center}

\noindent
The processes $N^a$ and $N^b$ having the same average intensities, taking expectations, we get the following result:
\begin{proposition}
The impact function $MI$ of a metaorder executed between $0$ and $\tau$ is:

\begin{equation}
\label{metaorderimpact}
MI(t):= \E[P_t-P_0]=F \int_0^{t\wedge\tau} \zeta(t-s)ds.
\end{equation}
\end{proposition}

\subsubsection{Apparent power law impact function}

Here, we consider that the Hawkes processes used to model the order flows are close to criticality (which corresponds to empirical measures, see Section \ref{alm}). We show that if the length of the metaorder $\tau$ is small compared to the correlation scale $1/(1-\int \phi )^{1/\alpha}$ that appears in Section \ref{alm}: $1\ll\tau\ll 1/(1-\int \phi )^{1/\alpha}$, then the renormalized impact function is close to a power law with exponent $1-\alpha$.\\

\noindent
In order to formally express and show this statement, we consider the same sequence of Hawkes processes as in Section \ref{alm}. Let us also denote $(\tau^T)_T$ the sequence of sizes of metaorders. In the same vein as before, we formally write that the length of metaorders is small compared to the correlation scale as:
\begin{assumption}
\label{metaordersize}
$\tau^T\rightarrow +\infty$ and $\tau^T(1-a_T)^{1/\alpha}\rightarrow 0.$
\end{assumption}

\noindent
Let us define the renormalized impact for $t\in[0,1]$: 
\begin{center}
$RMI^T(t)=\frac{1-a_T}{(\tau^T)^{1-\alpha}}MI^T(t\tau^T).$
\end{center}

\noindent
We have the following result:

\begin{theorem}
\label{theo3}
Under Assumption \ref{metaordersize}, the market impact is asymptotically a power law:
$$RMI^T(t)\rightarrow K' t^{1-\alpha},$$
with $K'=\frac{c^\alpha\kappa v \Phi}{1-\alpha}$.
\end{theorem}

\noindent
The proof is given in appendix.\\

\noindent
Of course, the renormalizing constant $(1-a_T)(\tau^T)^{\alpha-1}$ is chosen again so that the renormalized impact function converges.
This power law impact function is close to the ``square root law" which is empirically observed, see \cite{moroimpact}. Let us now see how this impact exponent $\nu=1-\alpha$ can be linked to the order flow persistence exponent $\gamma$.

\subsubsection{Link between the exponents}

For a price process implied by Equation \eqref{impactvolume} and for a Hawkes order flow with kernel whose norm is close to one and whose distribution tail is of order $1/x^{1+\alpha}$, we have shown that:

\begin{itemize}
\item{The exponent of the long memory of the order flow is $\gamma=1-2\alpha$, see Theorem \ref{theo1}.}
\item{The exponent of the impact power law is $\nu=1-\alpha$, see Theorem \ref{theo3}.}
\end{itemize}

\noindent
Therefore, we have: $$\nu=(1+\gamma)/2.$$

\noindent
On the data, $\gamma$ is stock dependent and seems to vary from $0.2$ to $0.7$, see \cite{propagator}, which in our framework implies that $\nu$ must vary between $0.6$ and $0.85$, which quite fits empirical data, see \cite{moroimpact}.

\begin{remark}
Theoretically retrieving the shape of the impact function from other market variables is an important problem in market microstructure. Three main possibilities have been suggested to solve it. The first one is to consider the risk aversion of investors and market makers as it is done in \cite{gabaix2006institutional}. Another way is to consider statistical models for the flow of market orders and price changes as in \cite{bacryhawkesorder}. Here we proceed somehow as in \cite{hmsdcsd}, \cite{efficiencyimpact} or \cite{tothanomalous} by introducing a martingale condition on the price which allows us to imply the impact function only from the order flow. Let us notice that the link between the exponent of the impact function and the exponent of the long memory of the order flow that we get is the same as in \cite{bacryhawkesorder} and \cite{hmsdcsd}.
\end{remark}

\section{Conclusion}
\label{conclusion}

We have recalled why the permanent market impact of a metaorder should be linear. Using this linearity as well as a martingale hypothesis on the price, we have computed an impact equation (Equation \eqref{impactvolume}) that allows us to retrieve price dynamics from any order flow model. We applied it to the example of order flows modelled by nearly unstable Hawkes processes. In this model, the sign of market orders presents apparent long memory and we have shown that we can recover many stylized facts of market impact. In particular, we computed a power law impact function whose exponent is linked to the long memory of the sign of market orders.\\

\noindent
The example of the Hawkes order flow allowed us to derive simple formulas on the future expectation of the volume. However, this model might be too simplistic. It would be interesting to consider more realistic frameworks such as the one presented in \cite{fragmentation} where the order flow is really built as a superposition of independent metaorders. In addition to their natural interpretation, these order flow models manage to link the size distribution of metaorders to the long memory property of the order flow.
However, such models are in a way much more complex and getting closed form prediction formulas for the order flow imbalance is probably very intricate.\\


\appendix

\section{Proofs}

\subsection{Proof of Proposition \ref{propogatheral}}

We proceed as in \cite{gatheralna}:\\

\noindent
For all $(v_1,v_2,T)>0$, we consider that on $[0,\theta T]$, we buy stock at a rate $v_1$ and on $[\theta T,T]$ sell stock at a rate $v_2$. We take $\theta=v_2/(v_1+v_2)$ so that this strategy is a round trip. It is shown in \cite{gatheralna} (Equation $(4)$) that the expectation of the cost of this strategy is:

\begin{center}
$E=v_1f(v_1)\int_0^{\theta T}\int_0^t G(t-s)dsdt+v_2f(v_2)\int_{\theta T}^T \int_{\theta T}^t G(t-s)dsdt-v_2f(v_1)\int_{\theta T}^T\int_0^{\theta T} G(t-s)dsdt$
\end{center}

\noindent
Writing $G=G_\infty+\tilde{G}$ with $\tilde{G}=G-G_\infty$:
\begin{eqnarray*}
\frac{E(v_1+v_2)^2}{T^2v_1v_2} & = & \frac{1}{2}G_\infty (f(v_1)v_2-f(v_2)v_1) \\
& + & \frac{(v_1+v_2)^2}{v_1v_2}\bigg(v_1f(v_1)\int_0^\theta\frac{1}{T}\int_0^{tT}\tilde{G}(s)dsdt+v_2f(v_2)\int_\theta^1 \frac{1}{T}\int_{0}^{T(t-\theta)}\tilde{G}(s)dsdt\\
& - &  v_2f(v_1) \int_\theta^1 \frac{1}{T}\int_{T(t-\theta)}^{Tt}\tilde{G}(s)dsdt\bigg).
\end{eqnarray*}

\noindent
Therefore, using Ces\`aro's lemma (applied to $\tilde{G}\rightarrow 0$), we show that only the first term does not tend to zero when $T$ tends to infinity and thus for $E$ not to be strictly negative when $T$ becomes large enough, we must have: $v_1f(v_2)\leq v_2f(v_1)$.\\

\noindent
Symmetrically, we get: $v_1f(v_2)= v_2f(v_1)$ therefore $f$ must be linear.

\subsection{Proof of Theorem \ref{theo1}}

Let us write:
$$\theta(x)=\int_0^{+\infty}\frac{e^{iu}}{u^x}du\text{ for }x\in]0,1[\text{ and }\theta(x)=\int_0^{+\infty}\frac{e^{iu}-1}{u^x}du\text{ for }x\in]1,2[.$$

\noindent
By definition of $X^T$, the covariance of its $h$-increments is linked to the covariance of the $Th$-increments of $N^T$ by\footnote{$C^{X,T}_{\tau,h}=\mathbb{E}[(X^T_{t+\tau+h}-X^T_{t+\tau})(X^T_{t+h}-X^T_{t})]-(\mathbb{E}[X^T_{t+h}-X^T_{t}])^2$.} $$C^{X,T}_{\tau,h}=A_T^2C^{N,T}_{T\tau,Th}.$$

\noindent
Let us now compute the Fourier transform of the correlation function of the $h$-increments of $X^T$. For a fixed $z\in \mathbb{R}^*$, we have
$$\hat{C}^{X,T}_{z,h}=\int_{-\infty}^{+\infty}C^{X,T}_{\tau,h}e^{iz\tau}d\tau=\frac{A_T^2}{T}\hat{C}^{N,T}_{\frac{z}{T},Th}.$$
Theorem 1 of \cite{bacryhawkesesti} yields

$$\hat{C}^{N,T}_{\frac{z}{T},Th}=Th \frac{\frac{\mu^T}{1-a_T} \hat{g}_{\frac{z}{T}}^{Th}}{|1-\hat{\phi}^T_\frac{z}{T}|^2}=Th \frac{\frac{\mu^T}{1-a_T} T\hat{g}_{z}^{h}}{|1-\hat{\phi}^T_\frac{z}{T}|^2}$$
with $g_{t}^u=(1-\frac{|t|}{u})^+$.\\

\noindent 
We now state the following technical lemma.
\begin{lemma}
\label{lemap}
For $z\in \mathbb{R}^*$
$$\hat{\phi}^T_\frac{z}{T}=a_T\Big[1+\theta (1+\alpha)\alpha \big(\frac{cz}{T}\big)^\alpha+ o(\frac{1}{T^\alpha})\Big].$$
\end{lemma}

\begin{proof}
The Fourier transform writes:
\begin{eqnarray*}
\hat{\phi}^T_\frac{z}{T}&=&\int_0^{+\infty} \phi^T(t)e^{it\frac{z}{T}}dt\\
&=&\int_0^{+\infty} \phi^T(t)dt+\int_0^{+\infty} \phi^T(t)(e^{it\frac{z}{T}}-1)dt\\
&=&a_T\big(1+T^{-\alpha}\int_0^{+\infty} \phi(\frac{Tu}{z})(e^{iu}-1)\frac{T^{1+\alpha}}{z}du\big).
\end{eqnarray*}

\noindent
However, for a fixed $u>0$, our hypothesis on the asymptotic behaviour of $\phi$ implies:$$\phi(\frac{Tu}{z})(e^{iu}-1)\frac{T^{1+\alpha}}{z}\rightarrow (e^{iu}-1) \frac{\alpha c^\alpha z^\alpha}{u^{1+\alpha}}.$$ Moreover, since $\phi$ is bounded, there exists $C>0$ such that $|\phi(x)|\leq \frac{C}{x^{1+\alpha}}$ and thus $$|\phi(\frac{Tu}{z})(e^{iu}-1)\frac{T^{1+\alpha}}{z}|\leq z^\alpha \frac{C}{u^{1+\alpha}}|e^{iu}-1|.$$ We can therefore apply the dominated convergence theorem to obtain the lemma.

\end{proof}

\noindent
Using Lemma \ref{lemap}, we get that

\begin{equation*}
\hat{C}^{N,T}_{\frac{z}{T},Th}=hT^2\frac{\frac{\mu^T}{1-a_T}}{|\theta(1+\alpha)|^2(\frac{c}{T})^{2\alpha}a_T^2 \alpha^2}\frac{\hat{g}_{z}^{h}}{|z|^{2\alpha}}\Big[1+o(1)+O\big(T^\alpha(1-a_T)\big)\Big].
\end{equation*}

\noindent
Replacing $\mu^T$ by its expression, leads to

\begin{equation*}
\displaystyle\lim_{T\rightarrow+\infty}\hat{C}^{X,T}_{z,h}=\frac{hC_\mu }{|\theta(1+\alpha)|^2c^{2\alpha} \alpha^2}\frac{\hat{g}_{z}^{h}}{|z|^{2\alpha}}.
\end{equation*}

\noindent
Let us now state a domination lemma for $\hat{C}^{X,T}_{z,h}$.
\begin{lemma}
There is a constant $C>0$ such that for all $z\in \mathbb{R}^*$
$$|\hat{C}^{X,T}_{z,h}|\leq C\min(\frac{1}{|z|^2},1)(1+\frac{1}{z^{2\alpha}}).$$
\end{lemma}
\begin{proof}
We begin the proof by noticing that there exist $c_1>0$ and $c_2>0$ such that  $$|\hat{g}_{z}^{h}|\leq c_1\min(\frac{1}{|z|^2},1)$$ and $$|1-\hat{\phi}^T_\frac{z}{T}|\geq c_2\min(\big|\frac{z}{T}\big|^\alpha,1).$$
If we now use these inequalities together with Theorem 1 of \cite{bacryhawkesesti} and the definition of $X^T$, we obtain
$$|\hat{C}^{X,T}_{z,h}|\leq  \frac{hC_\mu T^{-2\alpha}}{c_2^2 \min(\big|\frac{z}{T}\big|^\alpha,1)^2} c_1\min(\frac{1}{|z|^2},1)$$
which ends the proof.
\end{proof}

\noindent
The Fourier transform inversion formula writes $$C^{X,T}_{t,h}=\frac{1}{2\pi}\int_{-\infty}^{+\infty}\hat{C}^{X,T}_{z,h}e^{-itz}dz.$$
Thanks to the previous lemma, we can use the dominated convergence theorem to get that for a fixed $t$, $C^{X,T}_{t,h}$ has a limit when $T$ tends to infinity and
$$C^{X}_{t,h}:=\displaystyle\lim_{T\rightarrow+\infty}C^{X,T}_{t,h}=\frac{1}{2\pi}\int_{-\infty}^{+\infty}\frac{hC_\mu }{|\theta(1+\alpha)|^2c^{2\alpha} \alpha^2}\frac{\hat{g}_{z}^{h}}{|z|^{2\alpha}}e^{-itz}dz.$$
\noindent
We now state the following lemma which is given by a simple change of variables.
\begin{lemma}
The function whose Fourier transform is $|z|^{-2\alpha}$ is $\frac{1}{2Re\big(\theta(1-2\alpha)\big)|t|^{1-2\alpha}}$.
\end{lemma}

\noindent
Therefore, if we set $K=\frac{hC_\mu }{2Re\big(\theta(1-2\alpha)\big)|\theta(1+\alpha)|^2c^{2\alpha} \alpha^2}$, then the function $t\mapsto Kg_t^h\ast\frac{1}{|t|^{1-2\alpha}}$ has the same Fourier transform as $C^{X}_{t,h}$ and thus

\begin{equation}
\label{covar}
\displaystyle\lim_{T\rightarrow+\infty} C^{X,T}_{t,h}=Kg_t^h\ast\frac{1}{|t|^{1-2\alpha}}.
\end{equation}

\noindent
Computing the convolution ends the proof:
\begin{center}
$\displaystyle\lim_{T\rightarrow+\infty}C^{X,T}_{t,h}= K(|t+h|^{2H}+|t-h|^{2H}-2|t|^{2H})$ with $H=1/2+\alpha.$
\end{center}

\subsection{Proof of Proposition \ref{theo2}}

Starting from Equation \eqref{impactvolume}:
\begin{center}
$P_t=P_0+\kappa v \displaystyle\lim_{s\rightarrow+\infty}\E[N^a_s-N^b_s|\mathcal{F}_t]$.
\end{center}

\noindent
We rewrite $N^{a/b}_s=M^{a/b}_s+\int_0^s\lambda^{a/b}du$ which implies:
\begin{center}
$P_t=P_0+\kappa v \displaystyle\lim_{s\rightarrow+\infty}\E[M^a_s-M^b_s+\int_0^t \lambda^a_u-\lambda^b_udu+\int_t^s\lambda^a_u-\lambda^b_udu|\mathcal{F}_t]$.
\end{center}

\noindent
We use proposition \ref{formlambdamart} to replace $\lambda^{a/b}$ and we get:
\begin{center}
$P_t=P_0+\kappa v \displaystyle\lim_{s\rightarrow+\infty}\E[M^a_s-M^b_s+\int_0^t \lambda^a_u-\lambda^b_udu+\int_t^s \int_0^u\psi(u-x)(dM^a_x-dM^b_x)du|\mathcal{F}_t]$.
\end{center}

\noindent
Using that $M^{a/b}$ is a martingale and that:
\begin{center}
$dM_x=dN_x-\lambda_x dx=dN_x-(\mu+\int_0^t \phi(x-r)dN_r) dx,$
\end{center}
we get:
\begin{eqnarray*}
P_t&=&P_0+\kappa v [N^a_t-N^b_t+\int_t^{+\infty} \int_0^t\psi(u-x)(dN^a_x-dN^b_x)du\\
&+&\int_t^{+\infty} \int_0^t\psi(u-x)\int_0^x \phi(x-r)(dN^a_r- dN^b_r)dxdu].
\end{eqnarray*}

\noindent
Inverting the integrals in the previous equation, we get:

\begin{eqnarray*}
N^a_t-N^b_t & = & \int_0^t(dN^a_s-dN^b_s),
\end{eqnarray*}
and
\begin{eqnarray*}
\int_t^{\infty} \int_0^t\psi(u-x)(dN^a_x-dN^b_x)du & = & \int_0^t\int_t^{\infty} \psi(u-x) du (dN^a_x-dN^b_x)\\
& = & \int_0^t\int_{t-s}^{\infty}\psi(u')du' (dN^a_s-dN^b_s).
\end{eqnarray*}
Moreover
\begin{eqnarray*}
& &\int_t^{\infty} \int_0^t\psi(u-x)\int_0^x \phi(x-r)(dN^a_r- dN^b_r) dxdu \\
&=& \int_0^t\int_t^{\infty}\int_r^t \psi(u-x) \phi(x-r)dxdu (dN^a_r-dN^b_r) \\
& = & \int_0^t \int_{t-s}^{\infty}\int_0^{t-s} \psi(u'-x') \phi(x')dx'du' (dN^a_s-dN^b_s),
\end{eqnarray*}
where we have made the changes of variables: $u'=u-r$ and $x'=x-r$.

\noindent
Replacing those three terms in the previous equation ends the proof:
$$P_t=P_0+\kappa v \Big[\int_0^t (1+\int_{t-s}^{+\infty}\psi(u)du -\int_{t-s}^{+\infty}\int_0^{t-s} \psi(u-x) \phi(x)dxdu )(dN^a_s-dN^b_s)\Big].$$

\subsection{Proof of Theorem \ref{theo3}}

We start from the impact equation \eqref{metaorderimpact}, that we apply to our asymptotic:

\begin{eqnarray*}
RMI^T(t)&=&F \frac{1-a_T}{(\tau^T)^{1-\alpha}} \int_0^{t\tau^T} \zeta^T(s) ds\\
&=& F \frac{1-a_T}{(\tau^T)^{1-\alpha}} \int_0^t \frac{\kappa v \tau^T}{(1-a_T)(\tau^T)^\alpha} \Big[(1-a_T)(\tau^T)^\alpha+\int_{s\tau^T}^{+\infty} \phi^T(x)dx(\tau^T)^\alpha\Big]ds\\
&\rightarrow & \kappa vF \int_0^t\frac{c^\alpha}{s^{1-\alpha}}ds=\frac{c^\alpha\kappa v \Phi}{1-\alpha} t^{1-\alpha}.
\end{eqnarray*}

\noindent
Using the dominated convergence theorem with
$(1-a_T)(\tau^T)^\alpha\rightarrow 0$
and $\int_{s\tau^T}^{+\infty} \phi^T(x)dx\sim c^\alpha/\big((\tau^Ts)^\alpha\big)$
ends the proof.

\bibliographystyle{abbrv}
\bibliography{bibli}

\end{document}